\pdfoutput=1
\documentclass[natbib]{TOOLS/sigplanconf}
\usepackage{amsfonts,amsmath}

\newenvironment{proof}[1][Proof]{\begin{trivlist}
\item[\hskip \labelsep {\bfseries #1}]}{\end{trivlist}}
\usepackage{graphicx}
\usepackage{url}
\usepackage{verbatim}
\usepackage{color}
\definecolor{lgray}{gray}{0.92}
\definecolor{lblue}{rgb}{0.90,0.90,1.00}
\definecolor{lyellow}{rgb}{1.00,1.00,0.70}

\usepackage{listings}
\newenvironment{codex}{\small\verbatim}{\endverbatim\normalsize}

\lstnewenvironment{code}
    {\lstset{}%
      \csname lst@SetFirstLabel\endcsname}
    {\csname lst@SaveFirstLabel\endcsname}
\newtheorem{prop}{Proposition}

\newtheorem{df}{Definition}

\newcommand{\BI}[0]{\begin{itemize}}
\newcommand{\EI}[0]{\end{itemize}}
\newcommand{\I}[0]{\item}
\newcommand{\BE}[0]{\begin{enumerate}}
\newcommand{\EE}[0]{\end{enumerate}}

\newcommand{\BX}[0]{\begin{codex}}
\newcommand{\EX}[0]{\end{codex}}
\newcommand{\BV}[0]{\begin{verbatim}}
\newcommand{\EV}[0]{\end{verbatim}}

\def \bscale1 {0.25}
\def \bscale {0.25}
\def \N {\mathbb{N}}

\newcommand{\FIG}[4]{
\begin{figure}[htbp]
\centering
{\includegraphics[scale=#3]{figs/#4}}
\caption{#2}
\label{#1}
\end{figure}
}

\begin{document}
\conferenceinfo{} {}
\CopyrightYear{}
\copyrightdata{}

\titlebanner{}      
\preprintfooter{}   

\title{ 
  Binary Tree Arithmetic with Generalized Constructors
}
\subtitle{}
\authorinfo{~}
   {~}
   {~}
         

\maketitle

\begin{abstract}
We describe arithmetic computations in terms of operations on 
some well known free algebras (S1S, S2S
and ordered rooted binary trees) while emphasizing
the common structure present in all them
when seen as isomorphic with the set of
natural numbers.

Constructors and deconstructors seen through
an initial algebra semantics
are generalized to recursively defined
functions obeying similar laws.
Implementations using Scala's {\tt apply} and {\tt unapply}
are discussed together with
an application to a realistic arbitrary size
arithmetic package written in Scala, based on the
free algebra of rooted ordered binary trees,
which also supports rational number operations
through an extension to signed rationals
of the Calkin-Wilf bijection.

\category{D.3.3}
{PROGRAMMING LANGUAGES}{Language Constructs and Features}
[Data types and structures]

\terms
Algorithms, Languages, Theory

\keywords
arithmetic computations with free algebras,
generalized constructors,
declarative modeling of computational phenomena,
bijective G\"odel numberings and algebraic datatypes,
Calkin-Wilf bijection between natural and rational numbers
\end{abstract}

\section{Introduction}

Classical mathematics frequently uses functions defined on equivalence classes
(e.g. modular arithmetic, factor objects in algebraic structures) 
provided that it can prove that the choice of a representative
in the class is irrelevant. 

On the other hand, when working with proof assistants 
like  Coq \cite{Coq:manual}, based on type theory
and its computationally refined extensions like the
Calculus of Construction \cite{coquand:calculus:ic:88},
one cannot avoid noticing 
the prevalence of data types corresponding to
free objects, on top of which everything else is built
in the form of canonical representations.

Category-theory based descriptions of Peano arithmetic
fit naturally in the general view that data
types are {\em initial algebras} - in this case
the initial algebra generated by the successor function,
as a provider
of the canonical representation of natural
numbers.
Of course, a critical element in choosing such
free algebras is computational efficiency of 
the operations one wants to perform on them, in
terms of low time and space complexity.
For instance Coq formalizations of natural numbers
typically use binary representations while keeping
the Peano arithmetic view when more convenient 
in proofs \cite{Coq:manual}. 
Note also that free algebras corresponding
to one and two successor arithmetic (S1S and S2S)
have been used as a basis for decidable
weak arithmetic systems like \cite{buchi62} 
and \cite{Rabin69Decidability}.
It has been shown recently in
\cite{padl12,sac12} that the initial algebra of
ordered rooted binary trees 
corresponding to the language
of G\"odel's System {\bf T} types \cite{goedel1958bisher}
can be used as a the language of
arithmetic representations, with hyper-exponential
gains when handling numbers built
from ``towers of exponents'' like $2^{2^{... ^2}}$.
Independently, this view is confirmed by the
suggestion to use $\lambda$-terms as a form
of universal data compression tool \cite{complambda}
and by deriving bijective encodings of data types using a
game-based mechanism \cite{everybit}.

These results suggest a {\em free algebra based 
reconstruction of fundamental data types}
that are relevant as building blocks
for finite mathematics and computer science.
We will sketch in this paper an (elementary, not involving
category theory) foundation for
arithmetic computations with free algebras,
in which construction of sets, sequences,
graphs, etc. can be further carried out along the
lines of \cite{everything,ppdp10tarau,sac12}.

The paper is organized as follows.
We define in section \ref{free} isomorphisms between 
the free algebras of
signatures consisting of one constant 
and respectively,
of one successor ({\tt S}), 
two successors ({\tt O} and {\tt I}) 
and a free magma constructor ({\tt C}). 

To enable computations with the
objects of the free algebras,
we discuss the use of generalized
constructors / destructors derived from these
free algebras using the {\em apply / unapply}
constructs available in Scala (section \ref{gcons}).
As an application, a complete arbitrary size
rational arithmetic package using the
Calkin-Wilf bijection between positive
rationals and natural numbers is described in section \ref{rats}.
Sections \ref{related} and \ref{concl}
discuss related work and our conclusions.

\section{Free Algebras and Data Types} \label{free}

\begin{df}
Let  $\sigma$ be a signature consisting of an alphabet of constants (called generators) 
and an alphabet of function symbols (called {\em constructors}) with various arities. 
We define the free algebra $A_\sigma$ of signature $\sigma$ 
inductively as the smallest set such that: 
\BE
\I if $c$ is a constant of $\sigma$ then $c \in A_\sigma$
\I if $f$ is an n-argument function symbol of $\sigma$, then 
$\forall i, 0 \le i<n, t_i \in A_\sigma \Rightarrow f(t_0,\ldots,t_i,\ldots,t_{n-1}) \in  A_\sigma$. 
\EE
\end{df}
We will write {\tt c/0} for constants and {\tt f/n} for 
function symbols with {\tt n} arguments belonging to a given a signature.

More general definitions, e.g. as {\em initial objects}
in the category of algebraic structures,
are also used in the literature and a close relation exists with term algebras
distinguishing between function constructors 
(generating the {\em Herbrand Universe}) and predicate constructors
(generating the {\em Herbrand Base}).

Recursive data types in programming languages like Haskell, ML, Scala can be seen as a
notation for free algebras. We refer to \cite{wadler90unpub} 
for a clear and convincing description of this connection.

For instance, the Haskell declarations
\begin{codex}
data AlgU = U | S AlgU
data AlgB = B | O AlgB | I AlgB
data AlgT = T | C AlgT AlgT
\end{codex}
correspond, respectively to 
\begin{itemize}
\item the free algebra {\tt AlgU} with a single 
generator {\tt U} and unary constructor {S} 
(that can be seen as part of the language of Peano 
or Robinson arithmetic, or the decidable (W)S1S 
  system, \cite{buchi62})
\item the free algebra {\tt AlgB} with single generator {\tt B} and two unary 
  constructors {\tt O} and {\tt I} (corresponding to the language 
of the decidable system (W)S2S \cite{Rabin69Decidability}), 
as well as ``bijective base-2'' number notation \cite{wiki:bijbase}
\item the free algebra {\tt AlgT} with single 
generator {\tt T} and one binary constructor 
{\tt C} (essentially the same thing as the {\em free magma}
 generated by ${T}$).
\end{itemize}
The set-theoretical construction corresponding to the ``{\tt |}'' 
operation is {\em disjoint union} and 
the data types correspond to infinite sets generated by 
applying the respective constructors repeatedly.
The set-theoretical interpretation of ``self-reference'' in
such data type definitions can be seen as fixpoint operation on sets
of natural numbers as shown in the $P\omega$ construction
used by Dana Scott in defining the denotational semantics
for various $\lambda$-calculus constructs \cite{scott76}.

We will next ``instantiate'' some general results to make the underlying mathematics
as {\em elementary} and self-contained as possible. While category theory 
is frequently used as the mathematical backing for
data-types, we will provide here a simple set theory-based 
formalism, along the lines of \cite{bourbakiAlgI}.

We will start with the elementary mathematics behind the {\tt AlgT} data
type and follow with an outline for a similar 
treatment of {\tt AlgU} and {\tt AlgB}.

\subsection{The free magma of ordered rooted binary trees with empty leaves}

\begin{df}
A set $M$ with a (total) binary operation $\ast$ is called a {\em magma}.
\end{df}
\begin{df}
A morphism between two magmas $M$ and $M'$ is a function $f:M \to M'$ such that
$f(x \ast y)=f(x)\ast f(y)$.
\end{df}

Let $X$ be a set.
We define the sets $M_n(X)$ inductively 
as follows: $M_1(X)=X$ and for $n>1$, $M_n(X)$ 
is the disjoint union of the sets $M_k(X) \times M_{n-k}(X)$ for $0<k<n$.
Let $M(X)$ be the disjoint union of the family of sets $M_n(X)$ for $n>0$.
We identify each set $M_n$ with its canonical image in $M(X)$. 
Then for $w \in M_n(X)$,
we call $n$ the {\em length} of $w$ and denote it $l(w)$.
Let $w,w' \in M(X)$ and let $p=l(w)$ and $q=l(w')$. 
The image of $(w,w') \in M_p \times M_q$
under the canonical injection in $M(X)$ is called 
the composition of $w$ and $w'$ and
is denoted $w \ast w'$.

When $X=\{T\}$ where $T$ is interpreted as the ``empty'' leaf of 
ordered rooted binary trees,
the elements of $M_n$ can be seen as ordered rooted binary trees with
$n$ leaves while the composition operation ``$\ast$'' represents joining 
two trees at their roots to form a new tree.

\begin{df}
The set $M(X)$ with the composition operation $(w,w') \to w \ast w'$ is called
the {\em free magma} generated by $X$.
\end{df}

\begin{prop} \label{unique}
Let $M$ be a magma. Then every mapping $u:X \to M$ can be extended in a unique way to
a morphism of $M(X)$ into $M$.
\end{prop}
\begin{proof}
We define inductively the mappings $f_n:M_m(X) \to M$ as follows:
For $n=1, f_1=f$. 
For $n>1, \forall p \in \{1,..,n-1\}, f_n(w \ast w') = f_p(w) \ast f_{n-p}(w')$.
Let $g:M(X) \to M$ such that $\forall n>0,\forall x \in M_n(X), g(x)=f_n(x)$.
Then $g$ is the unique morphism of $M(X)$ into M which extends f.
\end{proof}

Note that this property corresponds of the {\em initial algebra} 
\cite{wadler90unpub} view of the
corresponding (ordered, rooted) {\em binary tree} data type.

\begin{df}
If $u:X \to Y$, we denote $M(u):M(X) \to M(Y)$ the unique morphism of 
magmas defined by the construction in {\bf Proposition \ref{unique}}.
\end{df}
If $v:Y \to Z$ then the morphism $M(v) \circ M(u)$ extends $v \circ u : X \to Z$ and 
therefore $M(v) \circ M(u) = M(v \circ u)$.

\begin{prop} \label{magmabij}
If $u:X \to Y$ is respectively injective, surjective, bijective then so is $M(u)$.
\end{prop}
It follows that
\begin{prop}
If $X=\{x\}$ and $Y=\{y\}$ and $u:X \to Y$ is 
the bijection such that $f(x)=y$, then
$M(u):M(X) \to M(Y)$ is a bijective 
morphism (i.e. an {\em isomorphism}) of free magmas.
\end{prop}
\begin{proof}
If $X$ is empty so is $M(X)$, hence $u$ is injective. 
If $u$ is injective, then $\exists u':Y \to X, u' \circ u = \mathit{id}_{M(X)}$ where
$\mathit{id}_{M(X)}$ denotes the identity mapping of $M(X)$. 
Then $M(u') \circ M(u) = M(u' \circ u) = \mathit{id}_{M(X)}$ and hence $M(u)$ is injective.
If $u$ is surjective, then $\exists u':Y \to X, u \circ u' = \mathit{id}_{M(Y)}$.
Then $M(u) \circ M(u') = M(u \circ u') = \mathit{id}_{M(Y)}$ 
and hence $M(u)$ is surjective.
If $u$ is bijective, than it is injective and surjective and so is $M(u)$.
\end{proof}

We will identify the data type {\tt AlgT} 
with the free magma generated by the set {\tt \{T\}}
and denote its binary operation $x \ast y$ as {\tt C x y}. 
It corresponds to the free
algebra (that we will also denote {\tt AlgT}) defined 
by the signature {\tt \{T/0, C/2\}}.

We can now instantiate the results described by the previous 
propositions to {\tt AlgT}:

\begin{prop} \label{magmaiso}
Let $X$ be an algebra defined by a constant $t$ and a binary operation $c$. 
Then there's a unique
morphism $f:~${\tt AlgT}$~\to X$ that verifies
\begin{equation} \label{zerofix}
f(T)=t
\end{equation}
\begin{equation}
f(C(x,y))=c(f(x),f(y))
\end{equation}
Moreover, if $X$ is a free algebra then $f$ is an isomorphism.
\end{prop}
\begin{proof}
It follows from Proposition \ref{magmabij} and equation $f(T)=t$, given
that $f$ is a bijection 
between the singleton sets {\tt\{T\}} and $\{t\}$.
\end{proof}
\subsection{The One Successor and Two Successors Free Algebras}
The {\em one successor} free algebra (also known as unary natural numbers 
or Peano algebra, as well as the language of the monoid $\{0\}^*$ and 
the decidable systems WS1S and S1S) is defined by the 
signature {\tt \{U/0, S/1\}}, where {\tt U} is a 
constant (seen as zero) and {\tt S} is the unary successor function. 
We will denote {\tt AlgU} this algebra and identify it with
its corresponding data type.

We state an analogue of Proposition \ref{magmaiso} 
for the free algebra {\tt AlgU}.

\begin{prop} \label{unaryiso}
Let $X$ be an algebra defined by a constant $u$ and a unary operation $s$. 
Then there's a unique morphism $f:~${\tt AlgU}$~\to X$ that verifies
\begin{equation} \label{zerofix}
f(U)=u
\end{equation}
\begin{equation}
f(S(x))=s(f(x))
\end{equation}
Moreover, if $X$ is a free algebra then $f$ is an isomorphism.
\end{prop}
Note that following the usual identification of data types and initial algebras,
{\tt AlgU} corresponds to the initial algebra ``{\bf 1 + \_}'' through the
operation $g=<${\tt U}$,${\tt S}$>$ seen as 
a bijection $g: 1+\mathbb{N} \to \mathbb{N}$.

The {\em two successor} free algebra (also known as 
bijective base-2 natural numbers or Peano algebra, 
as well as the language of the monoid $\{0,1\}^*$ and 
the decidable systems WS2S and S2S ) is defined by the 
signature {\tt \{B/0, O/1, I/1\} } where {\tt B} is a 
constant (seen as the empty sequence) and {\tt O, I} are two unary 
successor functions. We will denote {\tt AlgB} this 
algebra and identify it with
its corresponding  data type.

We can state an analogue of Proposition \ref{magmaiso} 
for the free algebra {\tt AlgB}.

\begin{prop} \label{binaryiso}
Let $X$ be an algebra defined by a constant $b$ 
and a two unary operations $o,i$. 
Then there's a unique morphism $f:~${\tt AlgB}$~\to X$ that verifies
\begin{equation} \label{zerofix}
f(B)=b
\end{equation}
\begin{equation}
f(O(x))=o(f(x))
\end{equation}
\begin{equation}
f(I(x))=i(f(x))
\end{equation}
Moreover, if $X$ is a free algebra then $f$ is an isomorphism.
\end{prop}

These observations suggest that for defining
isomorphisms between {\tt AlgU, AlgB} 
and {\tt AlgT} that enable
a complete set of equivalent arithmetic (and later set-theoretic)
operations on each of them, we will need a mechanism to
prove such equivalences. To this end, it will be enough
to prove that such non-constructor operations 
also form free algebras of matching signatures.

We will call {\tt terms} the
elements of our initial algebras.

\section{Generalized Constructors} \label{gcons}

The iso-functors supporting the equivalence between actual constructors
and their recursively defined function counterparts suggest
exploring programming language constructs that treat them
in a similar way. For instance it makes sense to
extend ``constructor-only benefits''
like pattern matching to their function counterparts.

Fortunately, 
constructors/deconstructors generalized to arbitrary functions 
are available in Scala
through {\tt apply/unapply} methods and in Haskell
through a special notation implementing {\em views},
under the implicit assumption that they define
inverse operations. 

One can immediately
notice that our free algebras provide
sufficient conditions under which
this assumption is enforced.
This suggests the possibility that
such generalized constructor/deconstructor pairs
could provide the combined
benefits of pattern matching and
data abstraction, with the implication that direct
syntactic support for such constructs
can bring significant expressiveness to
functional programming languages.

\subsection{Generalized Constructors 
with {\tt apply/unapply} in Scala} \label{gscala}

Besides supporting {\tt case classes} and {\tt case objects} that are used 
(among other things) to implement pattern matching,
Scala's {\tt apply} and {\tt unapply} methods 
\cite{scalaTypeCases,odersky2008programming} allow definition of
customized constructors and destructors (called {\em extractors} in Scala). 

We will next describe how arithmetic operations
with our {\tt AlgT} terms,
represented as ordered rooted binary trees,
can benefit from the use such ``generalized constructors''.

\lstset{backgroundcolor=\color{lyellow}}

Our {\tt AlgT} free algebra will correspond
in Scala to a {\tt case object / case class} definition,
combined with a mechanism to share actual code,
encapsulated in the {\tt AlgT} trait.

\begin{lstlisting}{language=Scala}
case object T extends AlgT
case class C(l: AlgT, r: AlgT) extends AlgT

trait AlgT {
  def s(z: AlgT): AlgT = z match {
    case T       => C(T, T)
    case C(T, y) => d(s(y))
    case z       => C(T, h(z))
  }
  
  def p(z: AlgT): AlgT = z match {
    case C(T, T) => T
    case C(T, y) => d(y)
    case z       => C(T, p(h(z)))
  }
\end{lstlisting}
Note the predecessor function called {\tt p} and
our auxiliary functions named {\tt d} (which ``doubles''
its input, assumed different from {\tt T}) and {\tt h} 
(which ``halves'' its input, assumed ``even'' 
and different from {\tt T}).
\begin{lstlisting}{language=Scala}
  def d(z: AlgT): AlgT = z match {
    case C(x, y) => C(s(x), y)
  }

  def h(z: AlgT): AlgT = z match {
    case C(x, y) => C(p(x), y)
  }
 }
\end{lstlisting}

We will define our {\em generalized constructor/destructor} {\tt S}
representing the successor function and predecessor function
on rooted ordered binary trees of type {\tt AlgT} by providing
{\tt apply} and {\tt unapply} methods expressed in terms of our
``real'' constructors {\tt T} and {\tt C} and the actual
algorithms defined in the (shared) trait {\tt AlgT}.

\begin{lstlisting}{language=Scala}
object S extends AlgT {
  def apply(x: AlgT) = s(x)

  def unapply(x: AlgT) = x match {
    case C(_, _) => Some(p(x))
    case T       => None
  }
}

\end{lstlisting}
The definition of the generalized constructor/destructor {\tt D}
representing double / half is similar. Note the use of the method
{\tt d} defined in the trait {\tt AlgT}.
\begin{lstlisting}{language=Scala}
object D extends AlgT {
  def apply(x: AlgT) = d(x)

  def unapply(x: AlgT) = x match {
    case C(C(_, _), _) => Some(h(x))
    case _             => None
  }
}
\end{lstlisting}
The definition of the generalized constructor/destructor {\tt O}
can be seen as corresponding to $\lambda x.2x+1$ and its inverse.
\begin{lstlisting}{language=Scala}
object O extends AlgT {
  def apply(x: AlgT) = C(T, x)

  def unapply(x: AlgT) = x match {
    case C(T, b) => Some(b)
    case _       => None
  }
}

\end{lstlisting}
The definition of the generalized constructor/destructor {\tt O}
can be seen as corresponding to $\lambda x.2x+2$ and its inverse.
Note the use of the generalized constructors
{\tt S}, {\tt D} and {\tt O}, both on the left and right side of {\tt match}
statements, illustrating their usefulness
both as constructors and as extractors.
\begin{lstlisting}{language=Scala}
object I extends AlgT {
  def apply(x: AlgT) = S(O(x))

  def unapply(x: AlgT) = x match {
    case D(a) => Some(p(a))
    case _    => None
  }
}
\end{lstlisting}
\subsection{A Scala-based Natural Number Arithmetic Package using {\tt AlgT} Terms}
We will now illustrate how the use of generalized constructors
helps writing a fairly complete set of arithmetic algorithms
on therms of {\tt AlgT} seen as natural numbers. 
For comparison purposes, the reader might want
to look at the Haskell code in \cite{sac12} where similar
algorithms are expressed using a type class-based mechanism.
However, while the use of type classes comes with the
benefits of {\em data abstraction} it needs separate
functions for constructing, deconstructing and recognizing
terms to express the equivalent of the
generalized constructors used here.

We start with a comparison function returning {\tt LT, EQ, GT}
and supporting a {\tt total order} relation on {\tt AlgT}, isomorphic to
the one on $\N$. Note here the use of the generalized constructors
{\tt O} and {\tt I} providing a view of the terms of {\tt AlgT}
as terms of the free algebra {\tt BinT}.
\begin{lstlisting}{language=Scala}
trait Tcompute extends AlgT {
  def cmp(u: AlgT, v: AlgT): Int = (u, v) match {
    case (T, T)       => EQ
    case (T, _)       => LT
    case (_, T)       => GT
    case (O(x), O(y)) => cmp(x, y)
    case (I(x), I(y)) => cmp(x, y)
    case (O(x), I(y)) => strengthen(cmp(x, y), LT)
    case (I(x), O(y)) => strengthen(cmp(x, y), GT)
  }
  
  val LT = -1
  val EQ = 0
  val GT = 1

  private def strengthen(rel: Int, from: Int) = 
    rel match {
      case EQ => from
      case _  => rel
    }
\end{lstlisting}
Addition is expressed compactly in terms of the generalized
constructors {\tt O}, {\tt I} and {\tt S}.
\begin{lstlisting}{language=Scala}
  def add(u: AlgT, v: AlgT): AlgT = (u, v) match {
    case (T, y)       => y
    case (x, T)       => x
    case (O(x), O(y)) => I(add(x, y))
    case (O(x), I(y)) => O(S(add(x, y)))
    case (I(x), O(y)) => O(S(add(x, y)))
    case (I(x), I(y)) => I(S(add(x, y)))
  }
\end{lstlisting}
The definition of subtraction is similar, except that the
code of the predecessor function {\tt p}
is conveniently inherited directly from the
trait {\tt AlgT}, given that the trait {\tt Tcompute}
extends it.
\begin{lstlisting}{language=Scala}
  def sub(u: AlgT, v: AlgT): AlgT = (u, v) match {
    case (x, T)       => x
    case (O(x), O(y)) => p(O(sub(x, y)))
    case (O(x), I(y)) => p(p(O(sub(x, y))))
    case (I(x), O(y)) => O(sub(x, y))
    case (I(x), I(y)) => p(O(sub(x, y)))
  }

\end{lstlisting}
The multiplication operation is similar to the Haskell
code in section \ref{comp}, except for the use of the
generalized constructor {\tt O}.
\begin{lstlisting}{language=Scala}
   def multiply(u: AlgT, v: AlgT): AlgT = (u, v) match {
    case (T, _) => T
    case (_, T) => T
    case (C(hx, tx), C(hy, ty)) => {
      val v = add(tx, ty)
      val z = p(O(multiply(tx, ty)))
      C(add(hx, hy), add(v, z))
    }
  }
\end{lstlisting}
Similarly, a constant time complexity definition is given here
for the exponent of 2 operation, by using the ``real'' constructor {\tt C}.
\begin{lstlisting}{language=Scala}
  def exp2(x: AlgT) = C(x, T)
\end{lstlisting}
The power operation {\tt pow} takes advantage of
the generalized constructors {\tt O} and {\tt I} on the
left side of a {\tt case} statement through the {\tt AlgB}
view of {\tt AlgT}.
\begin{lstlisting}{language=Scala}
  def pow(u: AlgT, v: AlgT): AlgT = (u, v) match {
    case (_, T)    => C(T, T)
    case (x, O(y)) => multiply(x, pow(multiply(x, x), y))
    case (x, I(y)) => {
      val xx = multiply(x, x)
      multiply(xx, pow(xx, y))
    }
  }
\end{lstlisting}
Efficient division with remainder is a slightly more
complex algorithm, where we take advantage of
generalized constructors, direct inheritance from
trait {\tt AlgT} as well as number of
previously defined functions:
\begin{lstlisting}{language=Scala}
  def div_and_rem(x: AlgT, y: AlgT): (AlgT, AlgT) =
    if (cmp(x, y) == LT) (T, x)
    else if (T == y) null // division by zero
    else {
      def try_to_double(x:AlgT, y:AlgT, k:AlgT): AlgT =
        if (cmp(x, y) == LT) p(k)
        else try_to_double(x, D(y), S(k))

      def divstep(n: AlgT, m: AlgT): (AlgT, AlgT) = {
        val q = try_to_double(n, m, T)
        val p = multiply(exp2(q), m)
        (q, sub(n, p))
      }
      val (qt, rm) = divstep(x, y)
      val (z, r) = div_and_rem(rm, y)
      val dv = add(exp2(qt), z)
      (dv, r)
    }
\end{lstlisting}
Division and reminder can be separated using Scala's projection functions:
\begin{lstlisting}{language=Scala}
  def divide(x: AlgT, y: AlgT) = div_and_rem(x, y)._1

  def reminder(x: AlgT, y: AlgT) = div_and_rem(x, y)._2
\end{lstlisting}
Finally, the greatest common divisor {\tt gcd} and the least common multiplier
{\tt lcm} are defined as follows:
\begin{lstlisting}{language=Scala}
  def gcd(x: AlgT, y: AlgT): AlgT = 
    if (y == T) x else gcd(y, reminder(x, y))

  def lcm(x: AlgT, y: AlgT): AlgT = 
    multiply(divide(x, gcd(x, y)), y)
}
\end{lstlisting}
The trait {\tt Tconvert} 
implements efficiently
conversion to/from Scala's {\tt BigInt}
arbitrary size integers using bit-level operations
corresponding to power of 2 and recognition
of odd and even natural numbers.
The function {\tt fromN} builds an {\tt AlgT}
tree representation equivalent to a {\tt BigInt}.
\begin{lstlisting}{language=Scala}
 trait Tconvert {
  def fromN(i: BigInt): AlgT = {
      def oddN(i: BigInt) =
        i.testBit(0)

      def evenN(i: BigInt) =
        i != BigInt(0) && !i.testBit(0)

      def hN(x: BigInt): BigInt =
        if (oddN(x))
          BigInt(0)
        else
          BigInt(1) + hN(x >> 1)

      def tN(x: BigInt): BigInt =
        if (oddN(x))
          (x - BigInt(1)) >> 1
        else
          tN(x >> 1)

    if (0 == i) T
    else C(fromN(hN(i)), fromN(tN(i)))
  }
\end{lstlisting}
The function {\tt toN} converts an {\tt AlgT}
tree representation to a {\tt BigInt}.
\begin{lstlisting}{language=Scala}
  def toN(z: AlgT): BigInt = z match {
    case T => 0
    case C(x, y) =>
      (BigInt(1) << toN(x).intValue()) *
        (BigInt(2) * toN(y) + 1)
  }
}
\end{lstlisting}
Note that for both these conversions we have used,
for efficiency reasons, the ``real constructors''
{\tt T} and {\tt C}, although much simpler
(and slower) converters can be built using either
the {\tt AlgB} or {\tt AlgU} view of {\tt AlgT} terms.

The use of Scala's generalized constructors
inspired by
our free algebra isomorphisms has shown the
combined flexibility of inheritance as a mechanism
for data abstraction and convenient pattern matching
allowing the design of our algorithms
in a functional style. The implicit use of
{\tt apply} and {\tt unapply} methods in combination
with our simple free algebra semantics
has facilitated the safe use of fairly complex
(mutually) recursive functions in the definition of
the generalized constructors.
The use of Scala's traits has facilitated
flexible inheritance mechanisms supporting
shared definitions without any additional syntactic
clutter.

\section{An Application: Rational Arithmetic in Scala with Calkin-Wilf Trees}
\label{rats}

We will extend our Scala code snippet 
described in subsection \ref{gscala}
to a realistic arbitrary size arithmetic package. 
It is somewhat unconventional, as it
is based on the Calkin-Wilf bijection \cite{ratsCW,rationals}
between $\N$ and the set of positive rational numbers $\mathbb{Q^+}$, 
rather than more typical representations like the arrays of long words
used in Java's {\tt BigDecimal} package, also 
adopted through a wrapper
class with the same name 
by Scala (which runs on top of the Java Virtual Machine).

Among its advantages, division (with non-zero) always
returns a finitely represented rational and ``no bit is lost''
in the representation as canonical rational numbers with
co-prime numerator/denominator pairs are bijectively
mapped to natural numbers. Our approach
emphasizes the fact that a mathematical concept
defined traditionally through equivalence
classes and quotients, can be expressed
entirely in terms of a free algebra-based
mechanism.

The trait {\tt Q} representing our rational number
data type  contains distinct constructors for
positive ({\tt P}), negative numbers ({\tt M}) and
zero ({\tt Z}).
\lstset{backgroundcolor=\color{lyellow}}
\begin{lstlisting}{language=Scala}
trait Q extends Qcode

case object Z extends Q
case class P(x: (AlgT, AlgT)) extends Q
case class M(x: (AlgT, AlgT)) extends Q
\end{lstlisting}
The actual code will be shared through the
trait {\tt Qcode} that also mixes-in
functionality from the natural number
operations defined in the traits {\tt Tcompute}
and {\tt Tconvert}.

We start with a type definition for ordered
pairs of natural numbers {\tt PQ} represented
as terms of {\tt AlgT} and
the conversion function to a
conventional fraction represented as
an ordered pair of {\tt BigInt} objects.
The conversion function {\tt toFraq} uses
the {\tt AlgT} to {\tt BigInt} converter {\tt toN}.
\begin{lstlisting}{language=Scala}
trait Qcode extends Tcompute with Tconvert {
  type PQ = (AlgT, AlgT)

  def toFraq(): (BigInt, BigInt) = this match {
    case Z         => (0, 1)
    case M((a, b)) => (-(toN(a)), toN(b))
    case P((a, b)) => (toN(a), toN(b))
  }
\end{lstlisting}
The function {\tt t2pq} splits its argument {\tt u}
seen as a natural number into its corresponding Calkin-Wilf rational, 
represented as a pair of positive natural numbers of type {\tt PQ}.
Note the use of our generalized constructors {\tt O} and {\tt I}
distinguishing between odd and even numbers. The algorithm uses
an encoding of the path in the Calkin-Wilf tree as a member
of {\tt AlgB}, where {\tt O} is interpreted as a command to take
the left branch
and {\tt I} is interpreted as a command to take
the right branch at a node of the Calkin-Wilf tree (shown
in Fig. \ref{cwtree}, for a few small positive rationals, represented
as conventional fractions).
\FIG{cwtree}{The Calkin-Wilf Tree}{0.60}{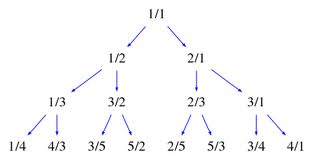}
\begin{lstlisting}{language=Scala}
  def t2pq(u: AlgT): PQ = u match {
    case T => (S(T), S(T))
    case O(n) => {
      val (x, y) = t2pq(n)
      (x, add(x, y))
    }
    case I(n) => {
      val (x, y) = t2pq(n)
      (add(x, y), y)
    }
  }
\end{lstlisting}
The function {\tt pq2t} fuses back into a ``natural number'' represented
as a term of {\tt AlgT}, corresponding to the path in the Calkin-Wilf tree,
a pair of co-prime natural numbers representing
the {\tt (numerator, denominator)} pair defining 
a positive rational number.
\begin{lstlisting}{language=Scala}
  def pq2t(uv: PQ): AlgT = uv match {
    case (O(T), O(T)) => T
    case (a, b) =>
      cmp(a, b) match {
        case GT => I(pq2t(sub(a, b), b))
        case LT => O(pq2t(a, sub(b, a)))
      }
  }
\end{lstlisting}
This brings us to the definition of the bijection between
{\em signed} rationals and terms seen through
the use of our generalized constructors {\tt O} and
{\tt I} as terms of {\tt AlgB} representing
natural numbers.
\begin{lstlisting}{language=Scala}
  def fromT(t: AlgT): Q = t match {
    case T    => Z // zero -> zero
    case O(x) => M(t2pq(x)) // odd -> negative
    case I(x) => P(t2pq(x)) // even -> positive
  }
\end{lstlisting}
Its inverse from signed rationals to terms of {\tt AlgT},
seen as natural numbers, proceeds by case analysis on the
{\tt Q} data type. Note that positive sign is encoded
by mapping to even naturals and negative sign is encoded
by mapping to odd naturals.
\begin{lstlisting}{language=Scala}
  def toT(q: Q): AlgT = q match {
    case Z    => T // zero -> zero
    case M(x) => O(pq2t(x)) // negative -> odd
    case P(x) => I(pq2t(x)) // positive -> even
  }
\end{lstlisting}
The bijection between Scala's {\tt BigInt}, seen
as a natural number type and signed rationals,
is defined as the pair of functions {\tt rat2nat}
and {\tt nat2rat}
\begin{lstlisting}{language=Scala}
  def nat2rat(n: BigInt): Q = fromT(fromN(n))
  
  def rat2nat(q: Q): BigInt = toN(toT(q))
\end{lstlisting}
Next we define a simplifier of
positive fractions represented as a pair,
to facilitate arithmetic operations
on our rationals.
\begin{lstlisting}{language=Scala}
   def pqsimpl(xy: PQ) = {
    val x = xy._1
    val y = xy._2
    val z = gcd(x, y)
    (divide(x, z), divide(y, z))
  }
\end{lstlisting}
We also use our simplifier to import and export non-canonically
represented rationals represented as {\tt BigInt} pairs.
\begin{lstlisting}{language=Scala}
  def fraq2pq(nd: (BigInt, BigInt)): PQ =
    pqsimpl((fromN(nd._1), fromN(nd._2)))

  def pq2fraq(nd: PQ): (BigInt, BigInt) =
    (toN(nd._1), toN(nd._2))
\end{lstlisting}
We are now ready for our arithmetic operations.
The template function {\tt pqop}, parameterized by a
function {\tt f}, will be shared
between addition and subtraction. Note that it
also involves simplification, to ensure
that the results are in a canonical
co-prime numerator/denominator form.
\begin{lstlisting}{language=Scala}
  def pqop(f: (AlgT, AlgT) => AlgT, xy:PQ, uv:PQ):PQ = {
    val (x, y) = xy
    val (u, v) = uv
    val z = gcd(y, v)
    val y1 = divide(y, z)
    val v1 = divide(v, z)
    val num = f(multiply(x, v1), multiply(u, y1))
    val den = multiply(z, multiply(y1, v1))
    pqsimpl((num, den))
  }
\end{lstlisting}
We can use it to define addition and subtraction
of positive rationals by simply instantiating our
function parameter {\tt f} to {\tt add} and {\tt sub}
operating on terms of {\tt AlgT}.
\begin{lstlisting}{language=Scala}
  def pqadd(a: PQ, b: PQ) = pqop(add, a, b)
  
  def pqsub(a: PQ, b: PQ) = pqop(sub, a, b)
\end{lstlisting}
The comparison operation providing a
total ordering of $\mathbb{Q^+}$ relies
on the function {\tt cmp} comparing terms of {\tt AlgT}
seen as natural numbers.
\begin{lstlisting}{language=Scala}
  def pqcmp(xy: PQ, uv: PQ) = {
    val (x, y) = xy
    val (u, v) = uv
    cmp(multiply(x, v), multiply(y, u))
  }
\end{lstlisting}
Multiplication, inverse and division on  
$\mathbb{Q^+}$ are defined as usual.
\begin{lstlisting}{language=Scala}  
  def pqmultiply(a: PQ, b: PQ) =
    pqsimpl(multiply(a._1, b._1), multiply(a._2, b._2))

  def pqinverse(a: PQ) = (a._2, a._1)

  def pqdivide(a: PQ, b: PQ) = pqmultiply(a, pqinverse(b))
\end{lstlisting}

We are ready to define 
arithmetic operations on the
set of signed rationals $\mathbb{Q}$,
by case analysis on
their sign. We start with the
opposite of a rational.
\begin{lstlisting}{language=Scala}
   def ropposite(x: Q) = x match {
    case Z    => Z
    case M(a) => P(a)
    case P(a) => M(a)

  }
\end{lstlisting}
Addition is defined by case analysis on
the sign and calls to the appropriate operations
on positive rationals.
\begin{lstlisting}{language=Scala}
  def radd(a: Q, b: Q): Q = (a, b) match {
    case (Z, y)       => y
    case (M(x), M(y)) => M(pqadd(x, y))
    case (P(x), P(y)) => P(pqadd(x, y))
    case (P(x), M(y)) => pqcmp(x, y) match {
      case LT => M(pqsub(y, x))
      case EQ => Z
      case GT => P(pqsub(x, y))
    }
    case (M(x), P(y)) => ropposite(radd(P(x), M(y)))
  }
\end{lstlisting}
Subtraction is defined similarly.
\begin{lstlisting}{language=Scala}
  def rsub(a: Q, b: Q) = radd(a, ropposite(b))

  def rmultiply(a: Q, b: Q): Q = (a, b) match {
    case (Z, _)       => Z
    case (_, Z)       => Z
    case (M(x), M(y)) => P(pqmultiply(x, y))
    case (M(x), P(y)) => M(pqmultiply(x, y))
    case (P(x), M(y)) => M(pqmultiply(x, y))
    case (P(x), P(y)) => P(pqmultiply(x, y))
  }
\end{lstlisting}
Finally we define the inverse on non-zero
rationals
\begin{lstlisting}{language=Scala}
  def rinverse(a: Q) = a match {
    case M(x) => M(pqinverse(x))
    case P(x) => P(pqinverse(x))
  }
\end{lstlisting}
and use it to derive from it a division operation on $\mathbb{Q}$
\begin{lstlisting}{language=Scala}
  def rdivide(a: Q, b: Q) =
    rmultiply(a, rinverse(b))
}
\end{lstlisting}
These operations conclude the trait {\tt Qcode}.
While this complete arithmetic package was built
mostly as a proof of concept for the expressiveness
of our free algebra based approach 
on progressively more interesting mathematical objects,
future work is planned for turning this package
into a practical tool. A first observation toward this end
is that, like in the case of Java's BigIntegers or the C-based
GMP package, one needs to use a hybrid approach, taking advantage of
actual machine words (64 bits at this point), to store
and operate on numbers that fit in a machine word.

\section{Related Work} \label{related}

Numeration systems on regular languages have been studied
recently, e.g. in \cite{Rigo2001469} and specific instances
of them are also known as 
bijective base-k numbers \cite{wiki:bijbase}.
Arithmetic packages similar to {\tt AlgU} and
{\tt AlgB} are part of libraries of proof assistants
like Coq \cite{Coq:manual} and the
corresponding regular 
languages have been used as a basis
of decidable arithmetic systems like
{\tt (W)S1S} \cite{buchi62} and {\tt (W)S2S} 
\cite{Rabin69Decidability}.

Arithmetic computations based
on more complex recursive data types like
the free magma of binary trees
(essentially isomorphic to the 
context-free language of balanced parentheses)
are described in \cite{sac12} and \cite{padl12},
where they are seen as G\"odel's {\tt System T} types,
as well as combinator application trees.
In \cite{ppdp10tarau}
a type class mechanism is used
to express computations on hereditarily
finite sets and hereditarily finite
functions.
However, none of these papers
provides proofs of the properties
of the underlying free algebras
or uses mechanisms similar to
the generalized constructors
described in this paper.

A very nice functional pearl \cite{rationals}
has explored in the past (using Haskell code)
algorithms related to the Calkin-Wilf
bijection \cite{ratsCW}. While using the same
underlying mathematics, our
Scala-based package works on terms of the
{\tt AlgT} free algebra
rather than conventional numbers,
and provides a complete package
of arbitrary size rational arithmetic operations
taking advantage of our
generalized constructors.

\section{Conclusion} \label{concl}

We have shown that free algebras corresponding to
some basic data types in programming languages
can be used for arithmetic computations
isomorphic to the usual
operations on $\N$.

As a new theoretical contribution, we have worked-out details of
proofs, based only on elementary mathematics,
of essential properties of the mutually recursive
successor and predecessor functions,
on the free algebra {\tt AlgT}
of ordered rooted binary trees.

A concept of {\em generalized constructor}, for which
we have found simple implementations
in {\tt Scala}, has been
introduced. By working in synergy with our
free algebra isomorphisms we have described,
using language constructs like Scala's {\tt apply / unapply},
simple and safe
means to combine data abstraction and pattern matching
in modern-day functional and object oriented languages.

As a new practical
contribution,
a complete arbitrary size signed rational number package written
in Scala has been derived working with terms of the {\tt AlgT}
free algebra of rooted ordered binary trees with empty leaves. 

Future work is planned to investigate
possible practical applications of our algorithms
to symbolic and/or arbitrary length integer arithmetic packages
and to parallel execution of arithmetic computations on {\tt AlgT}.

The code snippet showing the use of Scala's {\tt apply}
and {\tt unapply} methods to support generalized constructors
as well as the arithmetic on rationals
is available as a separate file
at http://\verb~######~.

\section*{Acknowledgement}
This research has been supported by NSF grant \verb~######~. 

\bibliographystyle{plainnat}
\bibliography{INCLUDES/theory,tarau,INCLUDES/proglang,INCLUDES/biblio,INCLUDES/syn}


\end{document}